\documentclass[preprint,nofootinbib]{revtex4-1}
\usepackage{amsmath}
\usepackage{amscd,amssymb,verbatim,comment}

\usepackage{hyperref}
\usepackage{pdfsync}

\newtheorem{Thm}{Theorem}
\newtheorem{Rem}{Remark}
\newtheorem{Prop}{Proposition}
\newtheorem{Lem}{Lemma}

\newenvironment{proof}{{\it Proof.}\ }{\hfill$\square$\par}

\newcommand{\n}{\nabla}

\newcommand{\CC}{\mathbb{C}}\newcommand{\RR}{\mathbb{R}}\newcommand{\ZZ}{\mathbb{Z}}
\newcommand{\Hol}{\operatorname{Hol}}
\newcommand{\Id}{\operatorname{Id}}
\newcommand{\Tr}{\operatorname{Tr}}

\newcommand{\Mat}{\operatorname{Mat}}
\newcommand{\IM}{\operatorname{Im}} \newcommand{\RE}{\operatorname{Re}}
\newcommand{\const}{\operatorname{const}}
\newcommand{\spec}{\operatorname{spec}}

\newcommand{\tilH}{\tilde{H}}
\newcommand{\hatH}{\hat{H}}

\newcommand{\calU}{\mathcal{U}}
\newcommand{\tilD}{\tilde{D}}
\newcommand{\hatD}{\hat{D}}
\newcommand{\detpm}{{\det}_\pm}
\newcommand{\tilB}{\tilde{B}}

\newcommand{\calD}{\mathcal{D}}
\newcommand{\tilT}{\tilde{T}}

\newcommand{\<}{\langle}\renewcommand{\>}{\rangle}

\begin{document}
\title{The Berry phase and the phase of the determinant}
\author{Maxim Braverman}
\thanks{Supported in part by the NSF grant DMS-1005888.}
\address{Department of Mathematics,
        Northeastern University,
        Boston, MA 02115,
        USA
         }

\begin{abstract}
We  show that under very general assumptions the adiabatic approximation of  the phase of the zeta-regularized determinant of the  imaginary-time Schr\"odinger operator with periodic Hamiltonian is equal to the Berry phase.  
\end{abstract}
\maketitle

\section{Introduction}\label{S:introduction}
In 1984 Michael Berry \cite{Berry84} discovered that an isolated eigenstate of an  adiabatically changing periodic Hamiltonian $H(t)$ acquires a phase, called the {\em Berry phase}. B.~Simon \cite{Simon83} gave an interpretation of this phase in terms of the holonomy of a certain Hermitian line bundle. We refer to \cite{GeomPhases_book89},\cite{GeomPhases_CLQ_book04} for further references and a detailed discussion of various aspects and applications of the Berry phase. 

It is known that in many interesting examples, \cite{Martinez89},\cite{AbanovAA02},\cite{BGSS07},\cite{BGSS09}, the Berry phase is related to the phase of the determinant of the corresponding imaginary-time Schr\"odinger operator $D_m= -i\frac{d}{dt}-imH(t)$ (here $m$ is a large constant). In this note we state and prove this relationship  under the most general assumptions about the Hamiltonian $H(t)$. 

Note that a regularization is needed to define the determinant of $D_m$ and the phase of the determinant depends of the choice of the regularization. To the best of our knowledge the study of this dependence in relation to the Berry phase was never conducted. In this note we consider the zeta-regularized determinant of $D_m$ and give a precise formulation and a rigorous proof of the relationship between the phase of this determinant and the Berry phase. In particular, we study the dependence of this relationship on the choice of the Agmon angle used in the definition of the zeta-function regularization, cf. Section~\ref{S:zeta}.

One of the difficulties in the computation of the phase of the determinant of $D_m$ is that the quantum adiabatic theorem, \cite{Kato50},\cite{Nenciu80}, does not hold for  the solutions of the imaginary time Schr\"odinger operator $D_m$ (cf. \cite{NenciuRasche92} for a discussion of the quantum adiabatic theorem for non self-adjoint operators). We explain this difficulty in more details in  Section~\ref{SS:plan}. 

The paper is organized as follows: In Section~\ref{S:Berry} we recall the definition of the Berry phase without assuming that the eigenvalues of $H(t)$ are isolated.  In Section~\ref{S:computationBerry} we collect some properties of the Berry phase. Most of these properties are well known to experts, but precise formulations and rigorous treatment of them, to the best of our knowledge, are missing in the literature.  In Section~\ref{S:zeta}  we recall the definition of the zeta-regularized determinant of elliptic operators. In Section~\ref{S:det=Berry} we formulate our main result -- Theorem~\ref{T:det=Berry}. In Section~\ref{S:prdet=Berry} we present  a proof of Theorem~\ref{T:det=Berry} based on the calculation of determinants of elliptic operators on a circle due to Burghelea, Friedlander and  Kappeler \cite{BFK91}.

\section{The Berry phase}\label{S:Berry}
In this section we fix the notation and recall the definition and the basic properties of the Berry phase. 

Let $H(t):\CC^N\to \CC^N$  be a family of self-adjoint Hamiltonians, which  depend smoothly on $t\in \RR$ and is $2\pi$-periodic $H(t)=H(t+2\pi)$. We view $H(t)$ as an operator-valued function on the circle $S^1=\{e^{it}:t\in [0,2\pi]\}$. Consider the Schr\"odinger equation 
\begin{equation}\label{E:Schrodinger}
	i\frac{d}{dt}\psi(t)\ = \ mH(t)\psi(t),
\end{equation}
where $m$ is a large real parameter.

\subsection{The case of an isolated eigenvalue}\label{SS:isolated}
Assume first that for each $t\in [0,2\pi]$ there exists an isolated non-degenerate eigenvalue $E(t)$ of $H(t)$ which depends continuously on $t$. The quantum adiabatic theorem \cite{Kato50},\cite{Nenciu80} claims that the solution $\psi_m(t)$ of the time-dependent Shr\"odinger equation \eqref{E:Schrodinger}  with initial value $\psi_m(0)=\phi_0$, where $H(0)\phi_0= E(0)\phi_0$, has the property that as $m\to\infty$, $\psi_m(t)$ approaches an eigenvector $\phi_t$ with $H(t)\phi_t=E(t)\phi_t$. More precisely, suppose that  $\phi_t\in \CC^N$ is a continuous  family of eigenvectors of $H(t)$ with eigenvalue $E(t)$,
\[
	H(t)\,\phi_t\ = \ E(t)\phi_t.
\]
Then as $m\to \infty$, 
\[
	\psi_m(t)\ = \ e^{i\alpha_m(t)}\phi_t\ +\ o(1),
\]
where $\alpha_m(t)\in \RR$.  In particular, it follows that 
\[
	\psi(2\pi)\ = \ e^{i\alpha_m(2\pi)}\phi(0) \ +\ o(1).
\]
Michael Berry \cite{Berry84} discovered that 
\begin{equation}\label{E:Berry isolated}
	\alpha_m(2\pi)\ = \ -m\int_0^{2\pi}E(t)\,dt \ + \ \gamma_E, 
\end{equation}
where $\gamma_E$ is independent of $m$. The number $\gamma_E$ is called the {\em Berry phase} corresponding to the energy level $E(t)$. 

\subsection{B.~Simon's description of the Berry phase}\label{SS:Simon construction}
Barry Simon \cite{Simon83} gave a geometric description of the Berry phase $\gamma_E$. Let
\[
	L^E\ :=\ \big\{(t,\psi):\,H(t)\psi =E(t)\psi\big\}
\]
be the complex line bundle over $S^1$ whose fiber $L^E_t$ over $e^{it}\in S^1$ is given by the eigenspace of $H(t)$ with eigenvalue $E(t)$. It has a natural unitary connection 
\begin{equation}\label{E:connectionline}
	\n \phi(t)\ = \ P_t\,\frac{d}{dt}\phi(t).
\end{equation}
Here $\phi(t)\in L^E_t\subset \CC^N$ is a section of $L$ and $P_t:\CC^N\to L^E_t$ is the orthogonal projection. Let $\Hol_\n:L^E_0\to L^E_0$ denote the holonomy of $\n$ along the circle $S^1$. Then $\Hol^E_\n$ is a multiplication by a complex number of absolute value one, which can be written as $e^{i\gamma_E}$. The number $\gamma_E\in \RR/2\pi\ZZ$ is exactly  the Berry phase corresponding to the energy level $E=E(t)$.

\subsection{The general case}\label{SS:nonisolated}

More generally, \cite{WilczekZee84}, suppose that  $\lambda(t)\in \RR$ is a continuous real valued function such that $\lambda(t)$ is not in the  spectrum of $H(t)$ for all $e^{it}\in S^1$.  Denote by $F_t^+\subset \CC^N$ (respectively $F_t^-$) the span of the eigenvectors of $H(t)$ corresponding to the eigenvalues which are bigger than $\lambda(t)$ (respectfively smaller than $\lambda$) and consider the vector bundles
\[
	F^\pm\ :=\ \big\{(t,\psi):\,\psi\in F_t^\pm\big\}
\]
over $S^1$.  Let $k= \dim{}F^-$ and consider the $k$-particle fermionic Fock space $\Lambda^k\CC^N$  (here $\Lambda^kV$ denotes the $k$-th exterior power of the vector space $V$). The  Hamiltonian $H(t)$ induces an operator 
\[
	H_k(t):\,\Lambda^k\CC^N \ \longrightarrow \ \Lambda^k\CC^N,
\]
whose smallest eigenvalue $E(t)$ is isolated and satisfies $E(t)= E_1(t)+\cdots+E_k(t)$, where $E_1(t),\dots,E_k(t)$ are all  the eigenvalues of $H(t)$ which are less than $\lambda(t)$, counted with their multiplicities.  The Berry phase $\gamma_{(-\infty,\lambda)}$ is defined to be the Berry phase corresponding to the energy level $E(t)$ of  the operator $H_k(t)$. 

If  for all $t\in S^1$ the spectrum of the restriction of $H(t)$ to $F(t)$ consists of simple eigenvalues $E_1(t),\ldots E_k(t)$ then 
\[
	\gamma_{(-\infty,\lambda)}\ =\ \gamma_{E_1}\ +\cdots +\ \gamma_{E_k}.
\] 
Notice, however, that $\gamma_{(-\infty,\lambda)}$ is defined  even when the eigenvalues $E_1(t),\ldots E_k(t)$ are not isolated and the individual Berry phases $\gamma_{E_i}$ are not well defined.  

The phase $\gamma_{(-\infty,\lambda)}$ can be interpreted as a holonomy of a connection in a way similar to the one presented in Section~\ref{SS:Simon construction}. As in \eqref{E:connectionline} we define a unitary connection on $F^-$ by 
\begin{equation}\label{E:connectionvector}
	\n \phi(t)\ = \ P_t\,\frac{d}{dt}\phi(t).
\end{equation}
where $P_t:\CC^N\to F_t^-$ denotes the unitary projection. Then
\begin{equation}\label{E:Berrylambda}
	e^{i\gamma_{(-\infty,\lambda)}} \ = \ \det \Hol_\n.
\end{equation}

\section{Computation of the Berry phase}\label{S:computationBerry}

In this section we present some basic fact about the Berry phase and give two explicit formulae for its computations. The results of this section are known to experts but the precise formulations and rigorous treatment of these results can not be easily found in the literature.  

Recall that $P_t:\CC^N\to F^-_t$ denotes the orthogonal projection on the space $F^-_t$ spanned by the eigenfunctions of $H(t)$ which have eigenvalues less than $\lambda(t)$. Our first aim is to construct a family of unitary matrices $U(t)$ such that with respect to the decomposition $\CC^N= F^+_0\oplus F^-_0$ the operator $U(t)^{-1}H(t)U(t)$ has  a block-diagonal form 
\[
	\begin{pmatrix}\tilH^+(t)&0\\0&\tilH^-(t)\end{pmatrix},
\]
where $\tilH^+(t)>\lambda(t)$ and $\tilH^-(t)< \lambda(t)$. Equivalently, $U(t)$ should satisfy $U(t)^{-1}P_tU(t)\ = P_0$.

\begin{Prop}\label{P:H+H-}
Let $U(t)\in \Mat_{N\times N}(\CC)$ denote the solution of  the initial value problem\footnote{The family $U(t)$ is sometimes referred to as {\em Kato's evolution}.}
\begin{equation}\label{E:dotPP}
 \begin{aligned}
  \frac{d}{dt}\,U(t) \ &= \ [\dot{P}_t,P_t]\,U(t)\\
  U(0)\ &= \ \Id.
 \end{aligned}
\end{equation}
Then for all $t\in [0,2\pi]$ we have
\begin{equation}\label{E:U-1PU}
	U(t)^{-1}P_t\,U(t)\ = \ P_0,
\end{equation}
and the Berry phase $\gamma_{(-\infty,\lambda)}$ is given by 
\begin{equation}\label{E:Hol=U}
	e^{\gamma_{(-\infty,\lambda)}}\ = \ \det\Hol_\n\ = \ \det\Big(P_0\circ U(2\pi)\circ P_0:\, F^-_0\ \to \ F^-_0\Big).
\end{equation}
\end{Prop}

\begin{proof}
Note that 
\begin{equation}\label{E:dotP}\notag
  \begin{aligned}
	\dot{P}_t\ = \ &\frac{d}{dt}\,P_t^2\ = \ \dot{P}_t\,P_t\ + \ P_t\,\dot{P}_t,\\
	P_t\,\dot{P}_t\,P_t\ &=\ P_t\,\big(\dot{P}_t-P_t\dot{P}_t\big)\ = \ 0.
  \end{aligned}
 \end{equation}
Using these equalities and \eqref{E:dotPP} we obtain (cf. for example, \cite{Nenciu80})
\begin{multline}\notag
	\frac{d}{dt}\,U(t)^{-1}P_t\,U(t)\ = \ U(t)^{-1}\,\Big(\, -[\dot{P}_t,P_t]\,P_t\ +\ \dot{P}_t\ +\ P_t\,[\dot{P_t},P_t]\,\Big)\,U(t) 
	\\ = \ 	
	U(t)^{-1}\,\Big(\, -\,\dot{P}_t\,P_t\ +\ P_t\,\dot{P}_t\,P_t\ +\ \dot{P}_t\ +\ P_t\,\dot{P}_t\,P_t\ - \ \dot{P}_t\,P_t\,\Big)\,U(t) \ = \ 0.
\end{multline}
Hence, $U(t)^{-1}P_tU(t)=\const=P_0.$ The equaltiy \eqref{E:U-1PU} is proven.

To prove \eqref{E:Hol=U} consider the solution $\Psi(t)$ of the initial value problem 
\begin{equation}\label{E:PUdotUP}
 \begin{aligned}
  \frac{d}{dt}\,\Psi(t) \ &= \ -\,P_0\,U(t)^{-1}\dot{U}(t)\,P_0\,\Psi(t)\\
  \Psi(0)\ &= \ P_0.
 \end{aligned}
\end{equation}
Note that  $\Psi(t)$ commutes with $P_0$ and its image lies in $F_0^-$.  In particular,
\begin{equation}\label{E:Psi=PPsi}
	\Psi(t)\ =\ P_0\Psi(t).
\end{equation}
Set 
\[
	\Phi(t)\ := \ U(t)\, \Psi(t).
\]
Then $\Phi(0)=P_0$ and 
\[
	\n \Phi(t)\ = \ P_t\,\frac{d}{dt}\,\big(\,U(t)\Psi(t)\,\big) \ = \ P_t\,\big(\, \dot{U}(t)\Psi(t)+U(t)\dot\Psi(t)\,\big).
\]
Using \eqref{E:PUdotUP} and \eqref{E:Psi=PPsi} we now obtain
\begin{equation}\notag
	\n \Phi(t)\ = \ UP_0U^{-1}\dot{U}\Psi+UP_0\dot\Psi \ = \ U\,P_0\,\Big(\, P_0\,U^{-1}\dot{U}P_0\Psi+\dot\Psi\,\Big) \ = \ 0.
\end{equation}
Hence, 
\[
	\Hol_\n\ =\  \Phi(2\pi).
\]
Since, $P_{2\pi}=P_0$ we obtain from \eqref{E:U-1PU} that $U(2\pi)P_0= P_0U(2\pi)P_0$ and
\[
	\Hol_\n\ =\  \Phi(2\pi)\ = \ U(2\pi)\Psi(2\pi)\ = \ P_0\,U(2\pi)P_0\Psi(2\pi).
\]

To finish the proof of \eqref{E:Hol=U}  it remains to show that $\det\Psi(2\pi)=1$. This follows from the following computation
\begin{multline}\label{E:detPsi=const}
	\frac{d}{dt}\,\log\det\Psi(t)\ = \ \Tr\,\dot\Psi\Psi^{-1} \ = \ -\Tr P_0U^{-1}\dot{U}P_0\\ = \ 
	-\Tr U^{-1}P_t\dot{U}U^{-1}P_tU\ = \ -\Tr P_t\dot{U}U^{-1}P_t\ = \ -\Tr P_t[\dot{P}_t,P_t]P_t\ = \ 0.
\end{multline}
\end{proof}

We now give a second formula for the Berry phase. Let $\calU(t):\CC^N\to \CC^N$ be a smooth family of unitary maps such that 
\footnote{
Such a family can be constructed, for example, as follows. By \eqref{E:U-1PU} the operator $U(2\pi)$ commutes with $P_0= P_{2\pi}$ and with $\Id-P_0$.  Hence, with respect to the decomposition $\CC^N=F_0^+\oplus{}F_0^-$ it has a block-diagonal form
\(
	U(2\pi)\ = \ \begin{pmatrix} U^+&0\\0&U^- \end{pmatrix}.
\)
Let $a^\pm$ be self-adjoint matrices such that $e^{2\pi ia^\pm}=U^\pm$. Then we can set 
\(
	\calU(t)\ := \  \begin{pmatrix} e^{-ita^+}&0\\0&e^{-ita^-} \end{pmatrix}\cdot U(t).
\)}

\begin{equation}\label{E:calU(2pi)}
	\calU(0)\ = \ \calU(2\pi)\ = \ \Id,
\end{equation}
and
\begin{equation}\label{E:calU-1PU}
	\calU(t)^{-1}P_t\,\calU(t)\ = \ P_0.
\end{equation}

Set $A(t):=\calU(t)^{-1}U(t)$. Then $A(2\pi)= U(2\pi)$ and  with respect to the decomposition $\CC^N=F_0^+\oplus{}F_0^-$ we have 
\begin{equation}\label{E:A=A+A-}
	A(t)\ = \ \begin{pmatrix} A^+(t)&0\\0&A^-(t) \end{pmatrix}.
\end{equation}
In particular, it follows from \eqref{E:Hol=U} that 
\begin{equation}\label{E:Hol=A}
	\Hol_\n\ =\ A^-(2\pi).
\end{equation}

\begin{Prop}\label{P:PcalUP}
The Berry phase $\gamma_{(-\infty,\lambda)}$ is given by 
\begin{equation}\label{E:Hol=calU}
	\gamma_{(-\infty,\lambda)}\ = \ i\int_0^{2\pi}\Tr\Big(P_0\,\calU(t)^{-1}\dot\calU(t)P_0\Big)\,dt.
\end{equation}
\end{Prop}
(Note that the operator $P_0\,\calU(t)^{-1}\dot\calU(t)P_0$ is skew-adjoint and, hence, the right hand side of \eqref{E:Hol=calU} is real.)

\noindent
\begin{proof}
By \eqref{E:A=A+A-}, the matrix $A(t)$ commutes with $P_0$.  Hence,
\begin{equation}\notag
	 P_0\,\calU(t)^{-1}\dot\calU(t)P_0 \ = \ P_0\,AU^{-1}\dot{U}A^{-1}P_0 \ - \ P_0\,\dot{A}A^{-1}P_0 \ = \ 
	A^-P_0U^{-1}\dot{U}P_0(A^-)^{-1}-\dot{A}^-(A^-)^{-1}.
\end{equation}
Recall from \eqref{E:detPsi=const} that $\Tr\big(P_0U^{-1}\dot{U}P_0\big)=0$. Hence, 
\[
	\Tr\Big(P_0\,\calU(t)^{-1}\dot\calU(t)P_0\Big) \ = \ -\Tr\dot{A}^-(A^-)^{-1} \ = \ -\frac{d}{dt}\log\det A^-.
\]
The proposition follows now from Proposition~\ref{P:H+H-} and \eqref{E:Hol=A}.
\end{proof}

\section{The zeta-regularized determinant}\label{S:zeta} 
Consider the operator $D_m=-i\frac{d}{dt}-imH(t)$ where $m$ is a large real parameter. We now recall the definition of the  zeta-regularized determinant of such an operator introduced by Ray and Singer \cite{RaySinger71}. For $\theta\in \RR$ denote $R_\theta:=\{\rho e^{i\theta}:\,\rho\ge0\}$. An angle $\theta\not\in \pi\ZZ$ is called an {\em Agmon angle} for $D_m$ if $R_\theta$ does not intersect the spectrum of $D_m$.  Any  $\lambda\in \CC\backslash{R_\theta}$ has a unique representation in the form 
$\lambda= |\lambda|\cdot e^{i\alpha} $ where $\theta<\alpha<\theta+2\pi$.  For $s\in \CC$ we set $\lambda_\theta^s:=  |\lambda|^s\cdot e^{is\alpha}$ and $\log_\theta(\lambda):= \log|\lambda|+i\alpha$.

Let $\lambda_1,\lambda_2,\ldots$ be the set of the eigenvalues of $D_m$ (each eigenvalue appear in this list the number of times equal to its algebraic multiplicity). The zeta-function of $D_m$ is defined by the formula
\[
	\zeta_{\theta,D_m}(s)\ := \ \sum_{j=1}^\infty (\lambda_j)_\theta^{-s}.
\]
The sum above is absolutely convergent for $\mathrm{Re}(s)>1$. Seeley  \cite{Seeley66} showed that it defines a holomorphic function on $\mathrm{Re}(s)>1$ which has a meromorphic extension to the whole complex plane which is regular at $0$. Notice that 
\[
	\zeta_{\theta,D_m}'(s)\ := \ -\sum_{j=1}^\infty \log_\theta\lambda_j\cdot(\lambda_j)_\theta^{-s}.
\]
Thus formally 
\[
	\zeta_{\theta,D_m}'(0)\ := \ -\sum_{j=1}^\infty \log\lambda_j
	\ = \ -\log_\theta\left(\lambda_1\cdot\lambda_2\cdots\right).
\]
Of course, the infinite sum and the infinite product in the equation above are divergent. However, this formal equality justifies the definition
\begin{equation}\label{E:zetadet}
	{\det}_\theta D_m \ := \ \exp\left(-\zeta_{\theta,D_m}'(0)\right).
\end{equation} 
The determinant $\det_\theta{}D_m$ depends on the choice of the Agmon angle $\theta$. However, if there are only finitely many eigenvalues of $D_m$ in the solid angle $\{\rho\cdot{}e^{i\alpha}:\theta_1\le \alpha\le\theta_2\}$ then $\det_{\theta_1}D_m= \det_{\theta_2}D_m$, cf. for example \cite[\S2.4]{BrAbanov}. Since the leading symbol of $D_m$ is self-adjoint, if  $0<\theta_1<\theta_2<\pi$ then there are only finitely many eigenvalues of $D_m$ in the solid angle  $\{\rho\cdot{}e^{i\alpha}:\theta_1\le \alpha\le\theta_2\}$, cf. \cite[\S10.1]{ShubinPDObook}. It follows that the determinant is the same for all Agmon angles  $\theta\in(0,\pi)$. We denote $\det_+D_m:= \det_\theta{}D_m$ for any $\theta\in (0,\pi)$. Similarly, the determinant does not depend on $\theta\in(-\pi,0)$ and we denote this determinant by $\det_-D_m$. 

The equation \eqref{E:zetadet} defines a particular choice of the logarithm of the determinant 
\[
	\log{\det}_\theta{}D_m\ =\  -\zeta_{\theta,D_m}'(0). 
\]
Remark, however, that $\zeta_{\theta,D_m}'(0)$ does depend on the angle $\theta\in (0,\pi)$. Thus $\log\det_\pm{}D_m$ is only defined modulo $2\pi i\ZZ$.

Let $D_m^*$ denote the adjoint of $D_m$. Because the zeta-regularized determinant depends on the choice of the Agmon angle $\theta$,  the complex conjugate of ${\det}_\theta{}D_m$ is equal not  to  ${\det}_\theta{}D_m^*$  but to  $\log{\det}_{-\theta}{}D_{m}^*$. Hence, 
\begin{equation}\label{E:complexconjugate}
	\overline{\log\detpm D_{m}}\ = \ \log{\det}_\mp{}D_{m}^*.
\end{equation}

\section{The main result}\label{S:det=Berry}
We are now ready to formulate our main result. 
\begin{Thm}\label{T:det=Berry}
Let $H(t):\CC^N\to \CC^N$ be a $2\pi$-periodic family of self-adjoint  Hamiltonians depending smoothly on $t$. Assume that 0 is not in the spectrum of $H(t)$ for all $t\in S^1$ and let $F^\pm_t\subset \CC^N$  denote the subspaces spanned by the eigenvectors of $H(t)$ with negative and positive eigenvalues respectively. Set $N^\pm=\dim{}F^\pm_t$. Then modulo $2\pi{}i\ZZ$ we have
\begin{align}
   \mathrm{Im} \log {\det}_+ D_m\ &= \ N^-\pi\ + \ \gamma_{(-\infty,0)}\ + \ o(1),\\
     \label{E:det=Berry-}
   \mathrm{Im} \log {\det}_- D_m\ &= \  N^+\pi\ + \  \gamma_{(-\infty,0)}\ + \ o(1),
\end{align}
where $\gamma_{(-\infty,0)}$ is the Berry phase defined in \eqref{E:Berrylambda}, and $o(1)\to 0$ as $m\to \infty$.
\end{Thm} 
 
\begin{Rem}{\em
In several interesting examples  (cf. \cite{BrAbanov}, \cite[\S6]{Br_symdet}) $ \mathrm{Im} \log {\det}_\pm D_m$ is independent of $m$, so that the $o(1)$ term in \eqref{E:det=Berry-} vanishes. It would be very interesting to find a general condition for this phenomenon. }
\end{Rem}


\section{Proof of Theorem~\ref{T:det=Berry}}\label{S:prdet=Berry}

\subsection{Burgelea-Friedlander-Kappeler formula}\label{SS:BFK}
Several steps in our proof are based on Theorem~1 of  \cite{BFK91} which gives a formula for  the determinant of a general elliptic operator on a circle in terms of its monodromy map. We only need the special case of this formula for  the operator of the type
\begin{equation}\label{E:calD}
	\calD\ := \ -i\frac{d}{dt}\ + \ A(t).
\end{equation}
 Let  $T(t)\in \Mat_{N\times N}(\CC)$ denote the solution of  the initial value problem 
\begin{equation}\label{E:T}
 \begin{aligned}
 	\calD\, T(t&) \ = \  0.\\
  	T(0)\ &= \ \Id.
 \end{aligned}
\end{equation}
The matrix $T(t)$ is called the {\em monodromy map} of the operator $\calD$.  Notice that the first equation in \eqref{E:T} is equivalent to 
\begin{equation}\label{E:T2}
    \frac{d}{dt}\,T(t) \ = \ -i\,A(t)\,T(t).
\end{equation}
As in \cite{BFK91} we set 
\[
	R(\calD)\ := \  \exp \Big( \frac{i}2\int_0^{2\pi}\Tr A(t)\,dt\Big)
\]
Define operators $\Gamma_\pm:\CC^N\to \CC^N$ by 
\[
	\Gamma_+\ := \ -\Id, \quad \Gamma_-\ := \ \Id
\]
and set
\[
	S_\pm(\calD)\ : = \ \det\Gamma_\pm\cdot  \exp \Big( \frac{i}2\int_0^{2\pi}\Tr\, (\Gamma_\pm A(t))\,dt\Big).
\]
Then 
\begin{equation}\label{E:SpmR}
\begin{aligned}
	S_+(\calD)\cdot R(\calD)\ &= \ (-1)^N,\\
	S_-(\calD)\cdot R(\calD)\ &= \ \exp \Big( i\int_0^{2\pi}\Tr A(t)\,dt\Big).
\end{aligned}
\end{equation}
By Theorem~1 of \cite{BFK91} 
\begin{equation}\label{E:BFK91a}\notag
	{\det}_\pm\calD\ = \ (-1)^NS_\pm(\calD)\cdot R(\calD)\det\big(\Id-T(2\pi)\big).
\end{equation}
Thus from \eqref{E:SpmR} we get 
\begin{equation}\label{E:detpm=calD}
\begin{aligned}
	{\det}_+\calD\ &= \ \det\big(\Id-T(2\pi)\big);\\
	{\det}_-\calD\ &= \ (-1)^N\exp \Big( i\int_0^{2\pi}\Tr A(t)\,dt\Big)\cdot\det\big(\Id-T(2\pi)\big).
\end{aligned}
\end{equation}

With this preliminaries discussed we are now ready to start the proof of Theorem~\ref{T:det=Berry}. We will give a brief outline of the proof in Section~\ref{SS:plan} after some additional notation is introduced.

\subsection{Bringing $H(t)$ to a blog-diagonal form}\label{SS:H+H-}
As in Section~\ref{S:Berry}, we let $F^\pm_t\subset \CC^N$  denote the subspaces spanned by the eigenvectors of $H(t)$ with negative and positive eigenvalues respectively.  We denote by  $P_t:\CC^N\to F^-_t$  the orthogonal projection.  Let $\calU(t)$ be the family of matrices which satisfy \eqref{E:calU(2pi)} and \eqref{E:calU-1PU}. Then 
\begin{equation}\notag
	\calU(t):F^\pm_0\ \to\ F^\pm_t.
\end{equation}
With respect to the decomposition $\CC^N= F^+_0\oplus{}F^-_0$ the operator 
\begin{equation}\label{E:tilH=}
	\tilH(t)\ :=\ \calU(t)^{-1}H(t)\,\calU(t)
\end{equation}
has a block-diagonal form 
\begin{equation}\label{E:U-1HU}
	\tilH(t)\ = \ \begin{pmatrix}\tilH^+(t)&0\\0&\tilH^-(t)\end{pmatrix}.
\end{equation}
Consider the operator
\begin{equation}\label{E:tilD}
	\tilD_m\ := \ \calU^{-1}D_m\,\calU \ = \ -i\frac{d}{dt}-im\tilH-i\,\calU^{-1}\dot\calU.
\end{equation}
Clearly, 
\begin{equation}\label{E:detD=dettilD}
	{\det}_\pm D_m\ = \ {\det}_\pm\, \tilD_m \ = \ 
	{\det}_\pm\, \Big(\, -i\frac{d}{dt}-im\tilH-i\,\calU^{-1}\dot\calU\,\Big).
\end{equation}

\subsection{The plan of the proof of Theorem~\ref{T:det=Berry}}\label{SS:plan}
Our proof of Theorem~\ref{T:det=Berry} is based on an application of \eqref{E:detpm=calD}. However it is not clear how to compute the large $m$ asymptotic of ${\det}\big(\Id-\tilT_m(2\pi)\big)$  when $\tilT_m(t)$ is the monodromy operator of $\tilD_m$.  The difficulty here is that  the imaginary time Schr\"odinger operator  $\tilD_m$ does not satisfy the quantum adiabatic theorem (cf. \cite{NenciuRasche92} for a discussion of an adiabatic limit for non self-adjoint Schr\"odinger operators). This, in particular, means that  as $m\to \infty$ the monodromy operator $\tilT_m(t)$ does not necessarily approach a block diagonal operator with respect to the decomposition $\CC^N= F^+_0\oplus{}F^-_0$.

Instead of applying \eqref{E:detpm=calD} directly  to $\tilD_m$, we first deform this operator, cf. Section~\ref{SS:tilDms}, in such a way that the phase of the determinant remains unchanged modulo $o(1)$ (note, however, that the absolute value of the determinant might change drastically under this deformation). Then in Lemma~\ref{L:detI-T} we compute the large $m$ asymptotic of the monodromy operator of the deformed operator.

\subsection{A deformation of the operator $\tilD_m$}\label{SS:tilDms}

Denote
\begin{equation}\label{E:hatD}
	\begin{aligned}
	\hatD_m\ := \ -i&\frac{d}{dt}-im\tilH-i\,P_0\,\calU^{-1}\dot\calU P_0 -i\,(\Id-P_0)\,\calU^{-1}\dot\calU\,(\Id-P_0),\\
	R\ &:= \ i\,P_0\,\calU^{-1}\dot\calU\,(\Id- P_0) +i\,(\Id-P_0)\,\calU^{-1}\dot\calU P_0,
	\end{aligned}
\end{equation}
and set
\begin{equation}\label{E:tilDms}
	\tilD_{m,s}\ : = \ \hatD_m\ - \ sR.
\end{equation}
In other words, using  the decomposition $\CC^N= F^+_0\oplus{}F^-_0$ we can write 
\[
	i\,\calU^{-1}\dot\calU \ = \ \begin{pmatrix}A^+&R_1\\R_2&A^-\end{pmatrix}.
\]
Then 
\begin{equation}\label{E:Dms decomposition}
  \begin{aligned}
	\hatD_m \ = \ \begin{pmatrix}\hatD_m^+&0\\0&\hatD_m^-\end{pmatrix}\ &= \  
	\begin{pmatrix}-i\frac{d}{dt}-im\tilH^+-A^+&0\\0&-i\frac{d}{dt}-im\tilH^--A^-\end{pmatrix};\\
	\tilD_{m,s}\ &= \ \begin{pmatrix}\hatD_m^+&-sR_1\\-sR_2&\hatD_m^-\end{pmatrix}.
 \end{aligned}
\end{equation}
Note also that $\tilD_{m,1}= \tilD_m$.

\begin{Lem}\label{L:tilDms invert}
There exists $m_0>0$ such that for all $m\ge m_0$, $s\in [0,1]$ the operator $\tilD_{m,s}$ is invertible.
\end{Lem}
\begin{proof}
Since for all $t\in \RR$,  zero is not in the spectrum of $H(t)$, there exists a constant $c >0$ such that  $\tilH^+(t)>c \cdot \Id_{F_0^+}$ and $\tilH^-(t)<-c \cdot\Id_{F_0^-}$. Hence,  for  every smooth  functions $\psi_+:S^1\to F_0^+$, $\psi_-:S^1\to  F_0^-$  we have
\begin{multline}\notag
	\Big|\,\<\,\hatD_m^\pm\psi_\pm\,\psi_\pm\> \,\Big|\ = \  \,\Big|\big\<\,\big(-i\frac{d}{dt}-A^\pm\big)\psi_\pm,\psi_\pm\,\big\> \ - \  im\,\big\<\,\tilH^\pm\psi_\pm,\psi_\pm\,\big\>\,\Big|
	\\ \ge \ \Big|\,m\,\big\<\,\tilH^\pm\psi_\pm,\psi_\pm\,\big\>\,\Big| \ \ge \ c \, m\,\|\psi_\pm\|^2,
\end{multline}
and 
\(
	\big\|\,\hatD_m^\pm\psi_\pm\,\big\|\ \ge \ c \,m\,\|\psi_\pm\|.
\)
It follows that for every $\psi=(\psi_+,\psi_-):S^1\to \CC^N$ 
\[
	\big\|\,\hatD_m\psi\,\big\|\ \ge \ c \,m\,\|\psi\|.
\]
Then for $m> \frac{\|R\|}c $, $s\in [0,1]$ the operator $\tilD_{m,s}=\hatD_m-sR$ is invertible.
\end{proof}
Lemma~\ref{L:tilDms invert} implies that for large $m$ the determinant of $\tilD_{m,s}$ is well defined. 

\begin{Lem}\label{L:tilD=hatD}
As $m\to \infty$ we have 
\begin{equation}\label{E:tilD=hatD}
	\IM\log\detpm\tilD_{m,s}\ = \ \IM\log\detpm\hatD_m \ +\ o(1).
\end{equation}
\end{Lem}

\begin{proof}
It suffices  to show that 
\begin{equation}\label{E:ddsIMlogdet}
	\frac{\partial}{\partial s}\,\IM\log\detpm\tilD_{m,s} \ = \ o(1).
\end{equation} 

Since  $\Tr R=0$ it follows from \eqref{E:detpm=calD} that 
\[
	\frac{\partial}{\partial s}\,\IM\log{\det}_+\tilD_{m,s} \ = \ \frac{\partial}{\partial s}\,\IM\log{\det}_-\tilD_{m,s}, 
\]
and similar equality holds for the adjoint operator $\tilD_{m,s}^*$. Hence, using \eqref{E:complexconjugate} we obtain
\begin{equation}\label{E:ddsIMlogdet=}
	\frac{\partial}{\partial s}\,\IM\log\detpm\tilD_{m,s} 
	\ = \  \frac1{2i}\frac{\partial}{\partial s}\, \Big(\,\log\detpm\tilD_{m,s}- \log\detpm\tilD_{m,s}^*\,\Big).
\end{equation}
Formally, the derivative $\frac{\partial}{\partial s}\log\detpm\tilD_{m,s}$ should be equal to the trace of the operator $\tilD_{m,s}^{-1}\,\frac{\partial}{\partial s} \tilD_{m,s}$. However, the later operator is not of trace class and some regularization using  analytic continuation, cf. Section~\ref{S:zeta} ,  is needed to compute the derivative of $\log\detpm\tilD_{m,s}$. However, no analytic continuation is needed to compute the right hand side of \eqref{E:ddsIMlogdet=} as we shall now explain.

We have
\begin{equation}\label{E:D-D*}
	\tilD_{m,s}^{-1}\ - \ (\tilD_{m,s}^*)^{-1}\ = \ \tilD_{m,s}^{-1}\,\big(\tilD_{m,s}^*-\tilD_{m,s}\big) (\tilD_{m,s}^*)^{-1} 
	\ = \ 2im\,\tilD_{m,s}^{-1}\,\tilH\,(\tilD_{m,s}^*)^{-1} 
\end{equation}
Also,  since $\calU$ is unitary, the operator $i\,\calU^{-1}\dot\calU$ is self-adjoint, and, hence,  so is $R$. Thus
\begin{equation}\label{E:ddsD=ddsD*}
	\frac{\partial}{\partial s} \tilD_{m,s}^*  \ = \ \frac{\partial}{\partial s} \tilD_{m,s} \ = \ -R,
\end{equation}
From \eqref{E:D-D*} and \eqref{E:ddsD=ddsD*} we see that 
\[
	\tilD_{m,s}^{-1}\,\frac{\partial}{\partial s} \tilD_{m,s}\ - \ (\tilD_{m,s}^*)^{-1}\,\frac{\partial}{\partial s} \tilD_{m,s}^*
	\ = \ -2im\,\tilD_{m,s}^{-1}\,\tilH\,(\tilD_{m,s}^*)^{-1}R
\]
is an elliptic pseudo-differential operator of order -2 and, hence, is of trace class. A verbatim repetition of the argument in the proof of Proposition~1.3 of \cite{Forman87} shows now that 
\begin{multline}\label{E:dds IM}
	\frac{\partial}{\partial s}\,\IM\log\detpm\tilD_{m,s} \\ = \  
	\frac1{2i}\Tr  \Big[\, \tilD_{m,s}^{-1}\,\frac{\partial}{\partial s} \tilD_{m,s}\ - \ (\tilD_{m,s}^*)^{-1}\,\frac{\partial}{\partial s} \tilD_{m,s}^*\,\Big]
	 \ = \   -m\,\Tr\tilD_{m,s}^{-1}\,\tilH\,(\tilD_{m,s}^*)^{-1}R.
\end{multline}

The operator 
\[
	B_m\ := \ \tilD_{m,s}^{-1}\,\tilH\,(\tilD_{m,s}^*)^{-1}R
\] 
is a pseudo-differential operator with parameter $m$ of order -2, cf. \cite{ShubinPDObook}. Its leading symbol with parameter (cf.  \cite{ShubinPDObook}) is the same as the leading symbol of the operator
\[
	\tilB_m\ :=\   \big(-i\frac{d}{dt}-im\tilH\,\big)^{-1}\tilH\,	\big( -i\frac{d}{dt}+im\tilH\,\big)^{-1}R.
\]
Thus $B_m-\tilB_m$ is a differential operator with parameter of order -3. This means that its full symbol with parameter $\sigma(t,\xi,m)$ satisfies  
\[
	\big|\,\sigma(t,\xi,m)\big| \ \le \ C\,\big(1+|\xi|+m\big)^{-3}
\] 
for some constant $C>0$. Hence,
\begin{equation}\label{E:TrB-B0}
	|\Tr(B-\tilB_m)| \ = \ \Big|\, \int_0^{2\pi}\int_\RR\,  \sigma(t,\xi,m)\,d\xi dt\,\Big| \ \le \ C_1m^{-2}.
\end{equation}
Notice also that with respect to the decomposition $\CC^N=F^+_0\oplus F^-_0$ the operator $\tilB_m$ has the form 
\[
	\tilB\ = \ \begin{pmatrix} 0&\tilB_m^+\\\tilB_m^-&0\end{pmatrix}.
\]
Hence, $\Tr\tilB_m=0$. From \eqref{E:dds IM} and  \eqref{E:TrB-B0} we now conclude that 
\[
	\Big|\,\frac{\partial}{\partial s}\,\IM\log\detpm\tilD_{m,s}\,\Big| \ \le\ C_1m^{-1}.
\]
\end{proof}

\subsection{Computation of the determinant of $\hatD_m$}\label{SS:hatDm}
In view of Lemma~\ref{L:tilD=hatD}, to prove Theorem~\ref{T:det=Berry} it is enough to compute the phase of the determinant of the operator $\hatD_m$. Let  $T_m(t)\in \Mat_{N\times N}(\CC)$ denote monodromy map of the operator $\hatD_m$, cf. Subsection~\ref{SS:BFK}.

\begin{Lem}\label{L:BFK}
\begin{align}
 	{\det}_+ \hatD_m\  &= \ \det\big(\Id-T_m(2\pi)\big)    \label{E:BFK+}  \\
	{\det}_- \hatD_m\  &= \ (-1)^Ne^{m\int_{0}^{2\pi}\Tr H(t)dt} \det\big(\Id-T_m(2\pi)\big).
	\label{E:BFK-}
\end{align}

\end{Lem}
\begin{proof}
The equality \eqref{E:BFK+} is just the first equation in \eqref{E:detpm=calD}. To prove  \eqref{E:BFK-} set 
\begin{equation}\label{E:Am(t)}
	A_m(t)\ := \ -mi\tilH(t)-iP_0\,\calU^{-1}\dot\calU P_0 -i(\Id-P_0)\,\calU^{-1}\dot\calU\,(\Id-P_0).
\end{equation}
To use the second equation in \eqref{E:detpm=calD} we need to compute $\int_0^{2\pi}\Tr A_m(t)\,dt$. Notice, first, that  
\[
	\Tr\big(P_0\,\calU^{-1}\dot\calU P_0 +(\Id-P_0)\,\calU^{-1}\dot\calU\,(\Id-P_0)\big)\ = \ \Tr\,\calU^{-1}\dot\calU \ = \ 
	\frac{d}{dt}\log\det\calU.
\]
Hence, 
\begin{multline}\notag
	\exp \Big( i\int_0^{2\pi}\Tr\big(-iP_0\,\calU^{-1}\dot\calU P_0 -i(\Id-P_0)\,\calU^{-1}\dot\calU\,(\Id-P_0)\,\big)dt\Big) \\ = \ 
	\exp\big(\int_0^{2\pi}\frac{d}{dt}\log\det\calU\,dt\Big) \ = \ 
	\det\calU(2\pi)/\det\calU(0) \ = \ 1.
\end{multline}
Also by \eqref{E:tilH=} we have $\Tr \hatH(t)\ = \ \Tr{}H(t)$. Hence, from \eqref{E:Am(t)} and \eqref{E:detpm=calD} we obtain \eqref{E:BFK-}.
\end{proof}

%

\subsection{Computation of   $\operatorname{det}\big(\Id-T_m(2\pi)\big)$ }\label{SS:Tm(t)}
To finish the proof of Theorem~\ref{T:det=Berry} we now need to compute $\det\big(\Id-T_m(2\pi)\big)$.  From \eqref{E:U-1HU} we conclude that 
\begin{equation}\label{E:Tm=Tm+Tm-}
	T_m(t)\ = \ \begin{pmatrix} T_m^+(t)&0\\0&T_m^-(t)\end{pmatrix},
\end{equation}
where
\begin{equation}\label{E:Tmpm=}
  \begin{aligned}
	\dot{T}_m^+(t)\ &= \  - \Big(\,m\tilH^++(\Id-P_0)\,\calU^{-1}\dot\calU (\Id-P_0)\,\Big)\,T_m^+(t),\\
	\dot{T}_m^-(t)\ &= \  - \Big(\,m\tilH^-+P_0\,\calU^{-1}\dot\calU\,P_0\,\Big)\,T_m^-(t).
  \end{aligned}
\end{equation}
Hence, 
\begin{equation}\label{detI-T=Tpm}
	\det\big(\Id-T_m(2\pi)\big) \ = \ \det\big(\Id-T_m^+(2\pi)\big)\cdot \det\big(\Id-T_m^-(2\pi)\big).
\end{equation}

\begin{Lem}\label{L:detI-T}
As $m\to \infty$ we have
\begin{equation}\label{E:detI-Tpm}
	\det\big(\Id-T_m^+(2\pi)\big) =  1+o(1), \quad \det\big(\Id-T_m^-(2\pi)\big)  = (-1)^{N^-}\cdot \det T_m^-(2\pi)\cdot\big(1+o(1)\big).
\end{equation}
Hence, it follows from \eqref{detI-T=Tpm} that
\begin{equation}\label{E:detI-T}
	\det\big(\Id-T_m(2\pi)\big)   =  (-1)^{N^-}\cdot \det T_m^-(2\pi)\cdot\big(1+o(1)\big).
\end{equation}
\end{Lem}

\begin{proof}
Since, the operator $\calU^{-1}\dot\calU$ is skew-adjoint, for any $v\in \CC^N$ we have
\begin{equation}\label{E:ddtT}
	\frac{d}{dt}\, \|T_m^\pm(t)v\|^2 \ = \ 2\RE\, \big\<\,T_m^\pm(t)v,\dot{T}_m^\pm(t)v\,\big\> \ = \
	 -2\RE\, \big\<\,T_m^\pm(t)v,m\tilH^\pm(t)T_m^\pm(t)v\,\big\>.
\end{equation}
A sin the proof of Lemma~\ref{L:tilDms invert} there exists a constant $c>0$ such that  $\tilH^+(t)> c$, $\tilH^-(t)< -c$, for all $t\in [0,2\pi]$.   Hence, from \eqref{E:Tmpm=} and \eqref{E:ddtT} we obtain
\[
		\frac{d}{dt}\, \|T_m^+(t)v\|^2 \ \le \ -c m \|T_m^+(t)v\|^2, \qquad \frac{d}{dt}\, \|T_m^-(t)v\|^2 \ \ge \ cm  \|T_m^-(t)v\|^2.
\]
We conclude that  
\[
		\|T_m^+(t)v\|^2 \ \le \ e^{-cmt}\|v\|^2, \quad \|T_m^+(t)v\|^2 \ \ge \ e^{cmt}\|v\|^2, 
\]
and, hence, the spectrum $\spec(T_m^\pm(t))$ of the operators $T_m^\pm(t)$ satisfies
\begin{equation}\label{E:estimateT}
	\spec(T_m^+(t))\ \subset \ \big\{z\in \CC:\,|z|\le e^{-cmt/2}\big\},\quad \spec(T_m^-(t))\ \subset \ \big\{z\in \CC:\,|z|\ge e^{cmt/2}\big\}.
\end{equation}
The equality \eqref{E:detI-Tpm} follows immediately from \eqref{E:estimateT}.
\end{proof}

\begin{Lem}\label{L:detTm-}
For all $t\in [0,2\pi]$ the following equality holds modulo $2\pi\ZZ$
\begin{equation}\label{E:detTm-}
	\IM\log\det T_m^-(t)\ =\  i\int_0^t \Tr \Big(P_0\,\calU^{-1}(t)\,\dot{\calU}(t)P_0\Big)\,dt.
\end{equation}
\end{Lem}
\begin{proof}
The operator $\tilH^-(t)$ is self-adjoint, while the operator $P_0\,\calU^{-1}(t)\,\dot{\calU}(t)P_0$ is skew-adjoint. Hence $\Tr{}\tilH^-(t)$ is real and $\Tr\Big(P_0\,\calU^{-1}(t)\,\dot{\calU}(t)P_0\Big)$ is imaginary.
Using \eqref{E:Tmpm=} we obtain
\begin{multline}\notag
	\frac{d}{dt}\,\IM\log\det T_m^-(t)\ = \ \IM\,\Tr\, \dot{T}_m^-(t)\,\big(T_m^-(t)\big)^{-1}
	\\ = \ -\IM\Tr \Big(P_0\,\calU^{-1}(t)\,\dot{\calU}(t)P_0\Big)
	\ = \ i\Tr \Big(P_0\,\calU^{-1}(t)\,\dot{\calU}(t)P_0\Big).
\end{multline}

\end{proof}

\subsection{Proof of Theorem~\ref{T:det=Berry}}\label{SS:prdet=Berry}
Combining  \eqref{E:detD=dettilD} with Lemmas~\ref{L:tilD=hatD}, \ref{L:BFK}, \ref{L:detI-T}, and \ref{L:detTm-} we conclude that modulo $2\pi\ZZ$
\[
    \begin{aligned}
	\IM\log{\det}_+ D_m \ &= \ N^-\pi \ + \ i\Tr \Big(P_0\,\calU^{-1}(t)\,\dot{\calU}(t)P_0\Big) \ + \ o(1),\\
	\IM\log{\det}_- D_m \ &= \ N^+\pi\ +\  i\Tr \Big(P_0\,\calU^{-1}(t)\,\dot{\calU}(t)P_0\Big) \ + \ o(1).
   \end{aligned}
\]
Theorem~\ref{T:det=Berry} follows now from  Proposition~\ref{P:PcalUP}. \hfill$\square$

\section*{Acknowledgement}
I would like to thank Alexander Abanov from whom I learned about the relationship between the Berry phase and the determinants and who  suggested many important corrections and improvements to this paper.  

\end{document}